\documentclass{eptcs-modified}

\usepackage{amsmath}
\usepackage{amssymb}
\usepackage{url}
\usepackage{amsthm}
\usepackage{graphicx}
\usepackage{amscd}
\usepackage{stmaryrd}
\usepackage[usenames,dvipsnames,svgnames,table]{xcolor}
\usepackage{todonotes}
\usepackage{cleveref}
\usepackage{mdframed}
\usepackage{extarrows}
\usepackage{mathtools}
\usepackage{thm-restate}
\usepackage[T5,T1]{fontenc}

\begin{document}

\theoremstyle{definition}
\newtheorem{theorem}{Theorem}
\newtheorem{definition}[theorem]{Definition}
\newtheorem{problem}[theorem]{Problem}
\newtheorem{assumption}[theorem]{Assumption}
\newtheorem{corollary}[theorem]{Corollary}
\newtheorem{proposition}[theorem]{Proposition}
\newtheorem{example}[theorem]{Example}
\newtheorem{lemma}[theorem]{Lemma}
\newtheorem{claim}[theorem]{Claim}
\newtheorem{observation}[theorem]{Observation}
\newtheorem{fact}[theorem]{Fact}
\newtheorem{question}[theorem]{Open Question}
\newtheorem{conjecture}[theorem]{Conjecture}
\newtheorem{addendum}[theorem]{Addendum}
\newcommand{\uint}{{[0, 1]}}
\newcommand{\Cantor}{{\{0,1\}^\mathbb{N}}}
\newcommand{\name}[1]{\textsc{#1}}
\newcommand{\id}{\textrm{id}}
\newcommand{\dom}{\operatorname{dom}}
\newcommand{\Dom}{\operatorname{Dom}}
\newcommand{\codom}{\operatorname{CDom}}
\newcommand{\spec}{\operatorname{spec}}
\newcommand{\opti}{\operatorname{Opti}}
\newcommand{\optis}{\operatorname{Opti}_s}
\newcommand{\Baire}{{\mathbb{N}^\mathbb{N}}}
\newcommand{\hide}[1]{}
\newcommand{\mto}{\rightrightarrows}
\newcommand{\Sierp}{Sierpi\'nski }
\newcommand{\BC}{\mathcal{B}}
\newcommand{\C}{\textrm{C}}
\newcommand{\CC}{\textrm{CC}}
\newcommand{\UC}{\textrm{UC}}
\newcommand{\lpo}{\textrm{LPO}}
\newcommand{\llpo}{\textrm{LLPO}}
\newcommand{\leqW}{\leq_{\textrm{W}}}
\newcommand{\leW}{<_{\textrm{W}}}
\newcommand{\equivW}{\equiv_{\textrm{W}}}
\newcommand{\equivT}{\equiv_{\textrm{T}}}
\newcommand{\geqW}{\geq_{\textrm{W}}}
\newcommand{\pipeW}{|_{\textrm{W}}}
\newcommand{\nleqW}{\nleq_\textrm{W}}
\newcommand{\leqsW}{\leq_{\textrm{sW}}}
\newcommand{\equivsW}{\equiv_{\textrm{sW}}}
\newcommand{\Sort}{\operatorname{Sort}}
\newcommand{\aouc}{\mathrm{AoUC}}
\newcommand{\pitc}{\Pi^0_2\textrm{C}}
\newcommand{\powfin}{\mathcal{P}_{\textnormal{fin}}}
\newcommand{\Wei}{\mathfrak{W}}
\newcommand{\ptWei}{\mathfrak{W}_\bullet}
\newcommand{\eqpt}{=_\bullet}
\newcommand{\lept}{\le_\bullet}
\newcommand{\Nat}{\mathbb{N}}
\newcommand{\Var}{\mathcal{V}}
\newcommand{\partto}{\rightharpoonup}
\newcommand{\interpG}[1]{\llbracket #1 \rrbracket}
\newcommand{\interp}[1]{\llbracket #1 \rrbracket}
\newcommand\slice{\mathrm{Slice}}
\newcommand{\bnfeq}{\mathrel{::=}}
\newcommand{\bnfalt}{\; | \;}
\newcommand\cL{\mathcal{L}}
\newcommand\RegE[1]{RE_{#1}}
\newcommand\cS{\mathcal{S}}
\newcommand\simulatedBy{\cS}
\newcommand{\N}{\mathbb{N}}

\newcommand\cecilia[1]{\todo[inline,color=yellow!40]{Cécilia: #1}}

\newcommand\arno[1]{\todo[inline,color=green!40]{Arno: #1}}

\newcommand\eike[1]{\todo[inline,color=orange!40]{Eike: #1}}

\title{The equational theory of the Weihrauch lattice with multiplication}

\author{
Eike Neumann
\institute{Department of Computer Science\\Swansea University, Swansea, UK\\}
\email{e.f.neumann@swansea.ac.uk}
\and
Arno Pauly
\institute{Department of Computer Science\\Swansea University, Swansea, UK\\}
\email{Arno.M.Pauly@gmail.com}
\and
Cécilia Pradic
\institute{Department of Computer Science\\Swansea University, Swansea, UK\\}
\email{c.pradic@swansea.ac.uk}
}

\def\titlerunning{Theory of Weihrauch degrees}
\def\authorrunning{E.~Neumann, A.~Pauly \&  C.~Pradic}
\maketitle

\begin{abstract}
We study the equational theory of the Weihrauch lattice with multiplication, meaning the collection of equations between terms built from variables, the lattice operations $\sqcup,\sqcap$, the product $\times$, and the finite parallelization $^*$ which are true however we substitute Weihrauch degrees for the variables. We provide a combinatorial description of these in terms of a reducibility between finite graphs, and moreover, show that deciding which equations are true in this sense is complete for the third level of the polynomial hierarchy.
\end{abstract}

\subsubsection*{Erratum}
The first version of this preprint claimed that
our theory was not complete because $a \times (b \sqcap c) \le a \times (b \sqcap (c \times a))$ was not derivable without the pointedness axiom. This is false.
We have removed the erroneous statement and
added a proof that the axiomatization we propose is complete for Weihrauch degrees
if and only if it is complete for pointed degrees if $1 \le a$ is added (Theorem~\ref{thm:completeness-pointed-nonpointed-equiv}).

\section{Introduction}
The Weihrauch degrees $\mathfrak{W}$ come with a rich algebraic structure. Here, we consider the lattice operations $\sqcup, \sqcap$, the product $\times$ and the finite parallelization $^*$. They were the first operations on Weihrauch degrees studied in the literature \cite{brattka2,brattka3,paulyreducibilitylattice}. In order to better understand the structure of the Weihrauch degrees, we would like to characterize its equational theory, i.e.~we want to identify which equations between terms over the signature $(\mathfrak{W},\sqcap, \sqcup, \times, 1, (-)^*)$ are true for every instantiation of the variables by Weihrauch degrees. It was already observed in \cite{paulybrattka4} that the equational theory of $(\mathfrak{W},\sqcap, \sqcup)$ is the theory of distributive lattices, as $(\mathfrak{W},\sqcap, \sqcup)$ is a distributive lattice itself, and every countable distributive lattice embeds into $(\mathfrak{W},\sqcap, \sqcup)$ (via the Medvedev degrees).

\subsubsection*{Context}
Shortly after the precise definition of Weihrauch reducibility was proposed in \cite{gherardi}, Brattka and Gherardi raised the question whether the Weihrauch degrees $\mathfrak{W}$ form a Heyting or Brouwer algebra, in line with the overall idea that the structure of the Weihrauch degrees should be some kind of structure of constructive truth values. The question was answered in the negative in \cite{paulykojiro}; in particular, the Weihrauch degrees are not complete.

The broader question to what extend the Weihrauch degrees form an instance of already studied classes of structures linked to (constructive) logics remains an active area of research. While similarities and tentative connections are easy to find, a satisfactory answer is still eluding the community. A better understanding of the structure of the Weihrauch degrees seems integral to further progress, in particular for the hope to find the Weihrauch degrees to even be universal for some such logic.

\subsubsection*{Our contributions}

We start by investigating the equational theory of $(\mathfrak{W},\sqcap, \times)$. While we only conjecture that the axiomatization we propose for it is complete, we provide a combinatorial characterization in terms of reductions between finite graphs. This in turn will allow us to show that determining universal validity of equations in $(\mathfrak{W},\sqcap, \times)$ is $\Sigma^p_2$-complete. Since we only work with substructures that include $\sqcap$, we can write inequalities instead of equations by treating $t \le u$ as an abbreviation for $t = t \sqcap u$.

A Weihrauch degree is called \emph{pointed} if its representatives have some computable instance. 
We denote by $\ptWei \subset \Wei$ the subset of pointed degrees and write $\lept$ and $\eqpt$ instead of $\le$ and $=$ when we mean inequalities and equations that are only meant to be valid for pointed degrees. One easy observation is that if an (in)equality holds in $\Wei$, it also holds in $\ptWei$.

Occasionally a variant of Weihrauch reducibility called \emph{strong} Weihrauch reducibility is studied. Our operations are also operations on strong
Weihrauch degrees, which we shall denote by $\mathfrak{W}^s$, and it turns out that the equational theory of $(\mathfrak{W}^s,\sqcap,\times)$ and $(\mathfrak{W},\sqcap,\times)$ coincide. However, as the strong Weihrauch degrees are not distributive as a lattice \cite{damir}, this does not extend to signatures including $\sqcap$ and $\sqcup$.

\subsubsection*{Other work on the structure of the Weihrauch degrees}

Another investigation into the structure of the Weihrauch degrees is \cite{lmpsv}. A key result there is that $1$ is definable in just
$(\mathfrak{W},\leqW)$ -- but this involves multiple quantifiers, and thus is not directly relevant to the equational theory we study here.

At first glance, \cite{hertlingselivanov} seems to have a similar approach to our work. Hertling and Selivanov consider the complexity of a combinatorial reducibility which earlier work by Hertling \cite{hertling} tied to continuous Weihrauch reducibility. However, nothing seems to be known about definability of the fragment of the continuous Weihrauch degrees which is characterized by the combinatorial reducibility. 

\subsubsection*{Related work in proof theory}

We are not aware of axioms systems matching exactly those we provided for $(\ptWei, \sqcap, \times)$ in the literature.
Despite of this, there are tentative connections to be made between notions of correctness in substructural logics being phrased in terms of graphs
and our notion of reduction in terms of graphs. In particular, \cite{bgl} has a strikingly similar notion of reducibility between graphs generalizing boolean formulas.
this seems to indicate casting our axioms or our general notion of graph reducibility in a deep inference system might be of relevance in uncovering a complete axiomatization
for the theories we consider.

\section{Introducing the operations}
A general reference on Weihrauch degrees is the survey article \cite{survey-brattka-gherardi-pauly}. Here, we will recap briefly the operations we study. As we are interested in only the structure of the Weihrauch degrees, we can avoid introducing represented spaces, and instead work with the following:

\begin{definition}
Let $f,g : \subseteq \Cantor \mto \Cantor$ be multivalued functions. We write $f \leqW g$ and say that $f$ is \emph{Weihrauch reducible} to $g$ if there are computable $H : \dom(f) \to \dom(g)$ and (partial) $K : \Cantor \times \Cantor \partto \Cantor$ such that for all $x \in \dom(f)$ and $y \in g(H(x))$ we have that $K(x,y)$ is defined and $K(x,y) \in f(x)$.

The Weihrauch degrees are the equivalence classes for $\leqW$. We denote them by $\mathfrak{W}$. If $K$ does not depend on its first input, we have a strong Weihrauch reduction $f \leqsW g$. The strong Weihrauch degrees are denoted by $\mathfrak{W}^s$.
\end{definition}

Let $\langle \ \rangle : \Cantor \times \Cantor \to \Cantor$ be a standard pairing function. The operations we investigate are the following:

\begin{definition}
Given multivalued functions $f,g : \subseteq \Cantor \mto \Cantor$ , we define
\begin{itemize}
\item $f \sqcap g$ by $0p \in (f \sqcap g)(\langle q_0,q_1\rangle)$ if $p \in f(q_0)$ and  $1p \in (f \sqcap g)(\langle q_0,q_1\rangle)$ if $p \in g(q_1)$;
\item $f \sqcup g$ by $0p \in (f \sqcup g)(0q)$ if $p \in f(q)$ and $1p \in (f \sqcup g)(1q)$ if $p \in g(q)$;
\item $f \times g$ by $\langle p_0,p_1\rangle \in (f \times g)(\langle q_0,q_1\rangle$ if $p_0 \in f(q_0)$ and $p_1 \in g(q_1)$.
\end{itemize}
\end{definition}

These operations on multivalued functions induce operations on Weihrauch degrees. Moreover, the operations are monotone. The Weihrauch degrees are a lattice, with $\sqcap$ as meet and $\sqcup$ as join. 

We also make use of the constant $1$, which is the Weihrauch degree of $\operatorname{id} : \Cantor \to \Cantor$. The representatives of $1$ are exactly the computable multivalued functions whose domain contains a computable point. The upper cone of $1$ (i.e.~$\{f \mid 1 \leqW f\}$) is the class of pointed Weihrauch degrees $\ptWei$.

\begin{definition}
Given a multivalued function $f : \subseteq \Cantor \mto \Cantor$, we define $f^*$ by $0^\omega \in f^*(1p)$ and $\langle p_0,p_1\rangle \in f^*(0^{n+1}1\langle q_0, q_1\rangle)$ if $p_0 \in f(q_0)$ and $p_1 \in f^*(0^n1q_1)$.
\end{definition}

Plainly spoken, $f^*$ receives a number $n \in \mathbb{N}$ together with $n$ inputs for $f$ and outputs $n$ corresponding solutions.

\begin{figure}

\begin{center}
Basic properties of $(\Wei, \sqcap, \times, 1)$

\begin{mdframed}
\[\begin{array}{cr}
a \times (b \times c) = (a \times b) \times c \qquad a \times 1 = a & \text{monoid structure}\\
a \times b = b \times a & \text{commutativity} \\
a \le a \times a & \text{relevance}\\
a \sqcap b \le a \qquad a \sqcap b \le b  & \text{$\sqcap$ is a lower bound} \\
a \le b \; \wedge \; a \le c ~ \Rightarrow ~ a \le b \sqcap c & \text{$\sqcap$ is the greatest lb} \\
(a \times b) \sqcap c \le a \times (b \sqcap c)  & \text{half-distributivity}
\end{array}\]
\end{mdframed}

Additional axiom for the pointed case $(\ptWei, \sqcap, \times, 1)$
\begin{mdframed}
\[\begin{array}{cr}
1 \lept a & \text{bottom element}\\
\end{array}\]
\end{mdframed}
\end{center}
\caption{Our set of axioms for the systems under consideration, where we implicitly assume that equality is an equivalence relation, that all operators are monotone in each of their arguments and that inequalities are transitive.
Every axiom we state holds when substituting $\lept$ for $\le$ everywhere, and we take this as axioms for $\lept$ alongside the specific $\lept$ axioms.
}
\label{fig:axiom-base}
\end{figure}

At first, we investigate the structures $(\mathfrak{W},\sqcap,\times,1)$ and $(\ptWei, \sqcap, \times, 1)$. A partial axiomatization for these is provided in~\Cref{fig:axiom-base}. We do not explicitly list monotonicity of the operations and transitivity of $\leq$. To express the axioms, we will only need Horn clauses, that is formulas of the shape $\bigwedge_i t_i \le u_i \Rightarrow t' \le u'$. 
Our partial axiomatization essentially states that $\sqcap$ does behave like a greatest lower bound, that $\times$ and $1$ equip $\Wei$ with
a relevant commutative ordered monoid structure and that $\times$ \emph{half-distributes} over $\sqcap$.
Proofs that they hold in $\Wei$ are not difficult and can be found in \cite{paulybrattka4}.

We were unable to find a true inequality in $(\Wei, \sqcap, \times, 1)$ that does not follow from the stated axioms, and thus offer up the following:

\begin{conjecture}
$(\Wei, \sqcap, \times, 1)$ is fully axiomatized by~\Cref{fig:axiom-base}.
\end{conjecture}

We also show later that the axiomatization is complete for $\Wei$
if and only if it the pointed version is complete for $\ptWei$ (Theorem~\ref{thm:completeness-pointed-nonpointed-equiv}).

\section{A combinatorial notion of reducibility for terms}
\label{section:combinatorial}
The goal of this section is to give a combinatorial criterion to check the validity of
an inequality. Along the way we will see that the validity of an inequality where the semantics of the variables may range over all possible degrees can be reduced to the validity of the inequality where the variables are interpreted as sufficiently independent degrees. Let us first formally state the notion we investigate:

\begin{definition}
\label{def:univ-valid}
Let $u$ and $v$ be terms constructed from variables, the constant $1$ and the binary connectives $\sqcap, \times$. We say that an inequality $t \leq u$
  (respectively $t \lept u$) is \emph{universally valid} when for every interpretation of the variable as Weihrauch degrees (respectively pointed Weihrauch
  degrees), there is a corresponding Weihrauch reduction.
\end{definition}

Clearly, if $t \leq u$ is universally valid, then so is $t \lept u$. Conversely, if $t \lept u$ is universally valid, then $t \leq u$ is universally valid if and only if every variable occurring in $u$ also occurs in $t$. The reason for this is that the only difference between a Weihrauch reduction to $f$ and a Weihrauch reduction to $f \sqcup 1$ is that for the latter we do not necessarily need access to an instance for $f$. If every variable in $u$ appears also in $t$, then for every instantiation of the variables, we have available to us valid inputs for all Weihrauch degrees.

Now we prepare to define a combinatorial reduction between terms. To do so, we first need to introduce an interpretation of terms as graphs.

\begin{definition}
\label{def:cograph-interp}
A (finite coloured undirected) graph is a triple $(V, E, c)$ where $V$ is a finite set of vertices, $E \subseteq [V]^2$ is a set of edges and  $c : V \to \Var$ a colouring of the vertices by variables.
\end{definition}

\begin{definition}
Write $G + H$ for the disjoint union of two graphs and $G^\bot$ for $(V, [V]^2 \setminus E, c)$ when $G = (V, E, c)$.
The interpretation $\interpG{t}$ of terms $t$ as graphs is defined by induction as follows:
\begin{itemize}
\item for a variable $v$, $\interpG{v}$ is a graph with a single vertex with colour $v$
\item $\interpG{t \sqcap u} = \interpG{t} + \interpG{u}$ and
$\interpG{t \times u} = (\interpG{t}^\bot + \interpG{u}^\bot)^\bot$
\end{itemize}
\end{definition}

The graphs that arise as $\interpG{u}$ for some term $u$ are known as \emph{cographs}. They can also be characterized as those graphs that do not have an induced subgraph forming a line of length $3$ \cite{corneil}.
Below we give some examples of terms and corresponding cographs.

\begin{center}
\includegraphics{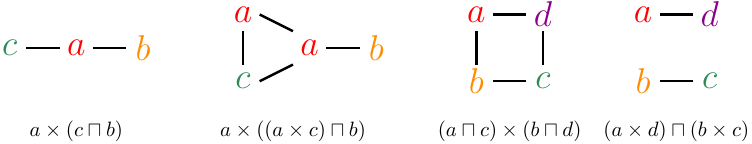}
\end{center}

The intuition is that a coloured graph $(V, E, c)$ can be read as the following
Weihrauch problem, assuming the colours correspond themselves to Weihrauch problems:
given for each vertex $v \in V$ an input $i_v$ to $c(v)$, output a maximal clique $C$ of $(V,E)$
together with solutions for every $i_v$s with $v \in C$. Under this interpretation, 
the problems denoted $\interpG{t}$ and $t$ are Weihrauch-equivalent.
This leads naturally to the following notion of combinatorial reduction between coloured graphs.

\begin{definition}
\label{def:comb-valid}
Given two graphs $G = (V_G, E_g, c_G)$ and $H = (V_H, E_H, c_H)$, call a partial function
$f : V_H \partto V_G$ a $\bullet$-\emph{reduction} of $G$ to $H$ if it respects colours (i.e., $c_H(v) = c_G(f(v))$ for all $v \in V_H$) and the image of every maximal clique in $H$ under $f$ contains a maximal clique of $G$. If $f$ is total, it is a \emph{reduction}.

An inequality $t \leq u$ (respectively $t \lept u$) is said to be \emph{combinatorially valid} if there is a reduction of $\interpG{t}$ into $\interpG{u}$ (respectively a $\bullet$-reduction).
\end{definition}

Let us illustrate the notion of combinatorial reduction on our examples.

\begin{center}
\includegraphics{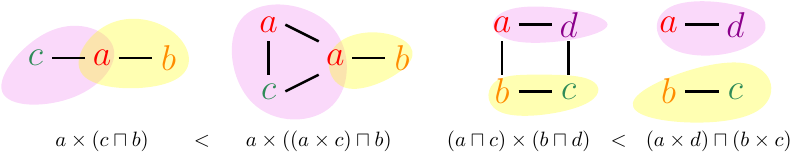}
\end{center}

In the two examples above, we have a strict reduction.
We do not specify the total
maps between vertices as they are uniquely determined by the colors in these
particular cases.
The shaded areas correspond to maximal cliques in the graphs and they have matching
colors when the reductions relate them.
There are no reductions the other ways around: for the first scenario, this is because
the clique $\{a,c\}$ on the left can only be mapped to a strict subclique of the triangular
clique on the right by a color-preserving map. For the second case, where we have a
unique color preserving-map, it is because $\{a,b\}$ does not contain a maximal clique
in the rightmost graph.

We are now ready to state the main theorem of this section:

\begin{theorem}
\label{lem:combred-valid}
Let $t$ and $u$ be terms constructed from variables and the binary connectives $\sqcap, \times$. Then inequalities $t \leq u$ and $t \lept u$ are
universally valid if and only if they are combinatorially valid.
\end{theorem}

To prove the theorem, we shall introduce a third notion of validity. Assume that we have a countable enumeration $(x_n)_{n \in \Nat} \in \Var^\Nat$ of all the variables that we may ever want to use in our terms.
We pick a countable strong antichain $(a_{n,m})_{(n,m) \in \Nat^2}$ of elements of $\Cantor$ in the Turing degrees.

\begin{definition}
\label{def:gen-valid}
We say that an inequality $t \leq u$ is \emph{generically valid} if it holds when we instantiate the $n$-th variable by the function $(a_{n,0}, k) \mapsto
  a_{n, k+1} : \{a_{n,0}\} \times \Nat \to \Cantor$. We say that $u \lept v$ is \emph{generically valid} if it holds when we instantiate the $n$-th variable by the function $k \mapsto a_{n, k} : \Nat \to \Cantor$.
\end{definition}

\begin{lemma}
\label{lem:generictocombinatorially}
A generically valid inequality $t \leq u$ or $t \lept u$ is also combinatorially valid.
\end{lemma}
\begin{proof}
That $t \leq u$ or $t \lept u$ respectively is generically valid is stating that there exists a Weihrauch reduction $f_t \leqW g_u$ between particular multivalued functions constructed from the terms. We fix an input to $f_t$ where different occurrences of the same variable receive different integer inputs. The valid solutions for $f$ are certain combinations of values $a_{i,j}$ from our chosen strong antichain. Following the construction of $\interpG{t}$, we see that the sufficient combinations of $a_{i,j}$ correspond to the maximal cliques in $\interpG{t}$. 

Since the reduction cannot obtain one of the $a_{i,j}$ unless the value is returned by $g_u$, we can obtain the reduction from $\interpG{t}$ to $\interpG{u}$ by observing which values provided by $g_u$ are the ones demanded by $f_t$.

If we are considering $t \lept u$, we are done. For $t \leq u$ we observe that the presence of $a_{n,0}$ in the domains means that $\interpG{u}$ cannot have any colours which are absent from $\interpG{t}$. We can thus make the partial reduction total by assigning some arbitrary vertex of the correct colour for any undefined output. 
\end{proof}

\begin{lemma}
\label{lem:combinatoriallytouniversal}
A combinatorially valid inequality $t \leq u$ or $t \lept u$ is also universally valid.
\end{lemma}
\begin{proof}
The forward reduction is entirely determined by the combinatorial one, and the backwards strong reduction can be obtained by taking a choice function between maximal cliques corresponding to the correctness criterion for the combinatorial reduction.
\end{proof}

\begin{proof}[Proof of \Cref{{lem:combred-valid}}]
It is relatively straightforward to prove
  \[ \text{universal} \xLongrightarrow{(1)} \text{generic} \xLongrightarrow{(2)} \text{combinatorial}
  \xLongrightarrow{(3)} \text{universal}\]

$(1)$ is trivial. $(2)$ is the statement of \Cref{lem:generictocombinatorially}.  $(3)$ is the statement of \Cref{lem:combinatoriallytouniversal}.
\end{proof}

\subsubsection*{Extending the picture to strong Weihrauch reducibility}
The proof of Lemma \ref{lem:combinatoriallytouniversal} proceeded by observing that a combinatorial reduction between terms translates to a Weihrauch reduction between the Weihrauch degrees obtained by substituting some arbitrary multivalued functions for the variables. The Weihrauch reductions we get are even strong Weihrauch reductions. Thus, the results of \Cref{{lem:combred-valid}} still hold when we replace universally valid by its counterpart for strong Weihrauch reducibility. This yields:

\begin{corollary}
$(\mathfrak{W},\sqcap,\times)$ and $(\mathfrak{W}^s,\sqcap,\times)$ have the same equational theory.
\end{corollary}

\subsubsection*{On the constant $1$}
Extending our results for inequalities $t \lept u$ built from $\sqcap,\times$ to terms also including the constant $1$ is straightforward, as our axioms for the pointed degrees allow us to simplify any term over $(\sqcap,\times,1)$ to produce either just $1$ or a term over $(\sqcap,\times)$.

For inequalities $t \leq u$, this is less straightforward. We can extend our interpretations of terms as coloured cographs to include $1$ by mapping $1$ to an uncoloured vertex. In the definition of a reduction map the uncoloured vertices are ignored other than in defining what the maximal cliques are. This means that the reduction map does not have to be defined on the uncoloured vertices, and the image of a maximal clique in  $\interpG{u}$ has to contain -- up to some uncoloured vertices -- a maximal clique in  $\interpG{t}$.

While having $1$ available to us makes it trivial to reduce universal validity of terms $t \lept u$ to universal validity of terms $t' \leq u'$, this is possible even without $1$:
\begin{proposition}
Let $t,u$ be terms built using the variables $x_0,\ldots,x_n$. The term $t \lept u$ is universally valid iff $t \times x_0 \times \ldots \times x_n \leq u \times x_0 \times \ldots \times x_n$ is universally valid.
\end{proposition}
\begin{proof}
Write $(V_t, E_t, c_t)$ and $(V_u, E_u, c_u)$ for $\interpG{t}$ and $\interpG{u}$ respectively,
and assume that we have that the vertices of $\interpG{x_0 \times \ldots \times x_n}$ being
$\{0, \ldots, n\}$ and disjoint from $V_t$ and $V_u$.
For both directions, we use the combinatorial validity of the reductions, which we know to be
equivalent to universal validity.

If we have a pointed combinatorial reduction $t \lept u$ witnessed by a partial map $f : V_u \partto V_t$,
then it can be extended to a non-pointed one witnessed by $\hat{f} : V_u \cup \{0, \ldots, n\} \to V_t \cup \{0, \ldots, n\}$ extending $f$ by setting $\hat{f}(i) = i$ for $0 \le i \le n$ and
$\hat{f}(v) = i$ when $v \in V_u \setminus \dom(f)$ and $c_t(v) = x_i$. It is easily checked to map maximal cliques to sets containing maximal cliques.

Conversely, if we have a combinatorial reduction $t \times x_0 \times \ldots \times x_n \leq u \times x_0 \times \ldots \times x_n$ witnessed by $g : V_u \cup \{0, \ldots, n\} \to V_t \cup \{0, \ldots, n\}$,
we may massage it into a partial map $\hat{g} : V_u \partto V_t$ witnessing a combinatorial reduction
$t \lept u$. $\hat{g}$ is defined as follows for $v \in V_u$:
\begin{itemize}
\item If $g(v) \in V_t$, we take $\hat{g}(v) = g(v)$
\item If $g(v) = i \in \{0, \ldots, n\}$ and $g(i) \in V_t$, we take $\hat{g}(v) = g(i)$
\item Otherwise $v \notin \dom(\hat{g})$
\end{itemize}
If we have a maximal clique $M$ in $\interpG{u}$, then $M \cup \{0, \ldots, n\}$ is a maximal clique
in $\interpG{u \times x_0 \times \ldots \times x_n}$ and thus gets mapped by $g$ to a set $S$ containing a maximal clique in $\interpG{t \times x_0 \times \ldots \times x_n}$ by $g$. Writing $S'$ for $S \cap V_t$, we clearly have that $S = S' \cup \{0, \ldots, n\}$ and that $S'$ contains a maximal clique of $\interpG{t}$. But then we also have $\hat{g}$ maps $M$ to $S'$: let us prove this by contradiction.
If an element $e$ of $S'$ is not in $\hat{g}(M)$, it means that there is no $v \in V_u$ such that $g(v) = e$ but there is some $i \in \{0, \ldots, n\}$ such that $g(i) = e$ by the first clause in the definition of $\hat{g}$. But since $i \in S$, there must be some $v'$ in $M$ such that $g(v') = i$. But then that would mean that $\hat{g}(v') = e$ by the second clause, a contradiction.
\end{proof}

\section{Completeness does not depend on pointedness}

\begin{theorem}
\label{thm:completeness-pointed-nonpointed-equiv}
The axioms of Figure~\ref{fig:axiom-base}
are complete for $(\Wei, \sqcap, \times, 1)$
if and only if they are complete
for $(\ptWei, \sqcap, \times, 1)$ when
adding $1 \lept a$
\end{theorem}

To prove this theorem, we first establish some
results for derivability that allows us to relate
the pointed and non-pointed theories. For any finite set
$X$ of variables, call $i_X$ the formula
$1 \sqcap \bigsqcap_{x \in X} x$.
The intuition is that the corresponding Weihrauch
problem requires at least an input for every $x$ to
be provided, but does not answer any of the corresponding
questions.

\begin{lemma}
\label{lem:contr-i-X}
For any finite sets of variables $X$ and $Y$,
we can derive $i_X \times i_Y = i_{X \cup Y}$.
\end{lemma}
\begin{proof}
Firstly we can show that we have $i_X \times i_X = i_X$.
This is because we have the relevance axiom
and $i_X \times i_X \le i_X \times 1 = i_X$,
since $i_X \le 1$ by definition.
Then, note that
$i_{X'} \le i_X$ whenever $X \subseteq X'$.
So we have $i_{X \cup Y} = i_{X \cup Y} \times i_{X \cup Y}
\le i_X \times i_Y$.
For the converse direction, it suffices to show
that $i_X \times i_Y \le 1$, $i_X \times i_Y \le x$ for
all $x \in X$ and, symmetrically, $i_X \times i_Y \le y$ for
all $y \in Y$. This is straightforward as $i_X \le 1$
and $i_X \le x$ for every $x \in X$.
\end{proof}

\begin{lemma}
\label{lem:der-tm-pt-nonpt}
For any term $t$ whose set of free variables is
$X$, we can derive $i_X \times t = t$.
\end{lemma}
\begin{proof}
Deriving $t \le i_X \times t$ is easy since $1 \le i_X$.
For the converse, we proceed by induction over $t$.
If $t = 1$, this amounts to $1 \times 1 = 1$.
If $t = t_1 \times t_2$, call $X_1$ and $X_2$ the
free variables of $t_1$ and $t_2$ respectively.
In this case we have $X = X_1 \cup X_2$,
so we can derive $i_X \times t \le i_{X_1} \times t_1 \times i_{X_2} \times
t_2$ using Lemma~\ref{lem:contr-i-X} and then conclude using
the inductive hypotheses and the monotonicity axiom for $\times$.
If we have $t = t_1 \sqcap t_2$, adopting similar naming
conventions, we also have $i_X \times t \le (i_{X_1} \times t_1) \sqcap (i_{X_2} \times t_2)$ using the fact that $i_X \le i_Y$ for any $Y \subseteq X$, and we similarly using monotonicity of
$\sqcap$.
\end{proof}

\begin{lemma}
\label{lem:der-pt-nonpt}
For any terms $t$ and $u$ such that all the free
variables of $t$ and $u$ are contained in a finite set $X$,
then we have that $i_X \times t \le u$ is
derivable if and only if $t \le_\bullet u$
is derivable.
\end{lemma}
\begin{proof}
The left-to-right direction is straightforward,
as we have $1 \le_\bullet x$ for all
$x \in X$ and therefore $1 \le_\bullet i_X$.
For the right-to-left direction, we need to reason
by induction on the size of a derivation $t \le_\bullet u$ and
make a case analysis on the last axiom applied. 
\begin{itemize}
\item If the last axiom applied is reflexivity and $t = u$,
we can simply apply Lemma~\ref{lem:der-tm-pt-nonpt} since
$X$ is a superset of the free variables of $t$.
\item If the last axiom applied is transitivity, i.e.
we have two subderivations of $t \le_\bullet v$
and $v \le_\bullet u$ for some $v$, we can assume,
without loss of generality, that all variables
occurring in $v$ occur in $X$.
Because if we had such a variable $a$ in $v$, we could
get derivations of the same size of $t = t[1/a] \le_\bullet v[1/a]$
and $v[1/a] \le_\bullet u[1/a] = u$.
Therefore, we can apply the induction hypothesis to
get derivations of $i_X \times u \le v$ and
$i_X \times v \le u$ which we can chain as follows
\[i_X \times t \le i_X \times (i_X \times t) \le i_X \times v \le u\]
\item If the last axiom applied is monotonicity of of $\times$,
i.e. we have derivations $t_1 \le_\bullet u_1$, $t_2 \le_\bullet u_2$
and $t = t_1 \times t_2$ and $u = u_1 \times u_2$, then we can
apply the inductive hypotheses to obtain
$i_X \times t_j \le u_j$ for $j \in \{1,2\}$ and
conclude by
\[i_X \times t_1 \times t_2 \le i_X \times t_1 \times i_X \times t_2 \le u_1 \times u_2\]
\item Monotonicity of $\sqcap$ is handled similarly.
\item All of the other axioms shared between the pointed
and non-pointed cases are handled trivially. So only
$1 \le_\bullet t$ remains, for which it suffices to prove
that $i_X \le t$. This is done via an easy induction over $t$.
\end{itemize}
\end{proof}

\begin{proof}[Proof of Theorem~\ref{thm:completeness-pointed-nonpointed-equiv}]
First let us assume that our theory is complete for $\ptWei$
and that we have that $t \le u$ combinatorially valid in $\Wei$.
In particular, our theory derives $t \lept u$. Via
combinatorial validity, we also know that
means that the free variables of $u$ are also free variables
of $t$; call $X$ the set of free variables of $t$. By
Lemma~\ref{lem:der-pt-nonpt}, we have that our theory
derives $i_X \times t \le u$, but we also have by
Lemma~\ref{lem:der-tm-pt-nonpt} that $t \le i_X \times t$,
so we can conclude that the theory is also complete for $\Wei$.

Conversely, if the theory is complete for $\Wei$ and
$t \le_\bullet u$ is combinatorially valid in $\ptWei$,
then, if $X$ is the set of free variables of $u$,
$i_X \times t \le u$ is easily seen to be combinatorially
valid; the partial map given by the pointed reduction
can be made total by mapping unmapped vertices to
those in $\interpG{i_X}$. Then, by completeness
we have $i_X \times t \le u$ derivable, and
since $1 \lept i_X$,
we a fortiori have $t = 1 \times t \lept i_X \times t \le u$
derivable.
\end{proof}

\section{Introducing joins and parallelization}

\begin{figure}

    \begin{center}

    Joins
    
    \begin{mdframed}
    \[\begin{array}{cr}
    a \le a \sqcup b \qquad b \le a \sqcup b  & \text{$\sqcup$ is an upper bound} \\
    b \le a \; \wedge \; c \le a ~ \Rightarrow ~ b \sqcup c \le a & \text{$\sqcup$ is the least ub} \\
    a \sqcap (b \sqcup c) = a \sqcap b \sqcup a \sqcap c  & \text{distributivity of $\sqcap$} \\
    a \times (b \sqcup c) = a \times b \sqcup a \times c  & \text{distributivity of $\times$} \\
    \end{array}\]
    
    \end{mdframed}
    
    Parallelization
    \begin{mdframed}
    \[\begin{array}{cr}
    a \le a^* \qquad (a^*)^* \le a^* & \text{closure operator}\\
    a^* \times a^* \le a^* & \text{duplicable}\\
    (a \sqcup b)^* = a^* \times b^* & \text{interaction with $\sqcup$}\\
    (a \sqcap b)^* = a^* \sqcap b^* & \text{interaction with $\sqcap$}\\
    (a \times b)^* = 1 \sqcup  a \times a^* \times b \times b^* & \text{interaction with $\times$}\\
    1^* = 1 & \text{interaction with $1$}\\
    \end{array}
    \]
    \end{mdframed}
    \end{center}
    \caption{Axioms for $\sqcup$ and $(-)^*$}
    \label{fig:axiom-ext}
    \end{figure}

We extend our axiomatization of $(\Wei, \sqcap, \times, 1)$ to an axiomatization of $(\Wei, \sqcap, \times, 1, \sqcup, (-)^*)$ as 
shown in Figure \ref{fig:axiom-ext}. 
This list of axioms is not free of redundancy. For example, duplicability of $(-)^*$ follows from $a = a \sqcup a$ and the interaction of $(-)^*$ with $\sqcup$.
We will show that a complete axiomatization of $(\ptWei, \sqcap, \times, 1)$ yields a complete axiomatization of
$(\ptWei, \sqcap, \times, 1, \sqcup, (-)^*)$.

Observe that the distributivity rules for $\sqcap$, $\times$, and $(-)^*$ in Figure \ref{fig:axiom-ext} allow us to rewrite every term $t$ as a join 
$t = \bigsqcup_i t_i$
where the $t_i$s are $\sqcup$-free terms such that parallelization is only applied to variables.
Distributing $\sqcap$ and $\times$ over $\sqcup$ can increase the size of the term exponentially.

However, we can re-write any term in linear time into one where parallelization is applied only to variables.
This is obvious in the pointed case, since the rule 
$(a \times b)^* = 1 \sqcup  a \times a^* \times b \times b^*$
together with pointedness yields 
$(a \times b)^* = a^* \times b^*$.
In general, an easy induction shows that any term of the form
    $t = \left(\bigtimes_{i = 1}^n t_i\right)^*$
is derivably equivalent to 
    $1 \sqcup \bigtimes_{i = 1}^n \left(t_i \times t_i^*\right)$
via a linear-size derivation.

First, we extend the combinatorial criterion to terms with the additional connectives. 
With our axioms, we can without loss of generality assume that parallelizations are only applied to variables; terms that do not satisfy this can be rewritten to terms that do so in linear time to
terms that are provably equivalent. We make this assumption below when defining the notion of combinatorial reduction for those terms.

A (coloured undirected) graph \emph{with parallelization information} is a tuple $(V,E,c,p)$ such that $(V,E,c)$ is a graph
as before and $p : V \to \{0,1\}$ is an additional colour information, meant to convey whether a particular vertex
$v$ denotes the variable $c(v)$ or the term $c(v)^*$. There is an obvious extension of $\interpG{-}$ that converts $\sqcup$-free
terms to graphs with parallelization information.

For a term $t$ over $(\ptWei, \sqcap, \times, \sqcup, (-)^*)$, we define the set $\slice(t)$ of \emph{slices of $t$} by induction as follows:
\begin{itemize}
\item If $t$ is $\sqcup$-free, then $\slice(t) = \{t\}$.
\item $\slice(t \sqcup u) = \slice(t) \cup \slice(u)$.
\item For $\otimes \in \{\sqcap, \times\}$, let 
\[\slice(t \otimes u) = \slice(t) \otimes \slice(u) := \left\{ t' \otimes u' \mid t' \in \slice(t), u' \in \slice(u)\right\}.\]
\end{itemize}

Any term $t' \in \slice(t)$ is $\sqcup$-free and can hence be associated with a graph with parallelization information $\interpG{t'}$.

\begin{definition}\label{def:comb-valid-ext}
    Let $t$ and $u$ be terms over $(\ptWei, \sqcap, \times, \sqcup, (-)^*)$.
A \emph{combinatorial reduction from} $t$ \emph{to} $u$ is given by
a map $s \colon \slice(t) \to \slice(u)$ and a family of relations 
$f_{t'} \subseteq V_{s(t')} \times V_{t'}$ 
indexed over the slices $t'$ of $t$, such that the following conditions are satisfied:
    \begin{enumerate}
        \item 
        For all $t' \in \slice(t)$, 
            if $p_{s(t')}(x) = 0$ then $(x,y), (x, y') \in f_{t'}$ implies $y = y'$ ($f_{t'}$ is single-valued on non-parallelized variables).
        \item 
        For all $t' \in \slice(t)$,
            if $(x, y) \in f_{t'}$ then $c_{t'}(y) = c_{s(t')}(x)$ ($f_{t'}$ respects colours).
        \item 
        For all $t' \in \slice(t)$,
            if $(x, y) \in f_{t'}$ then $p_{t'}(y) \leq p_{s(t')}(x)$ ($f_{t'}$ respects parallelization information).
        \item For all 
        $t' \in \slice(t)$ 
        and all maximal cliques 
        $M$
        in $\interpG{s(t')}$
        the image of $M$ under $f_{t'}$ 
        contains a maximal clique of $t'$.
    \end{enumerate}

We say that a pointed inequality $t \lept u$ is \emph{combinatorially valid} if there exists a corresponding combinatorial reduction between the two terms.
\end{definition}

\begin{lemma}\label{lem: combinatorial and universal validity for terms with join and *}
Combinatorial and universal validity coincide for terms over $(\ptWei, \sqcap, \sqcup, \times, (-)^*)$.
\end{lemma}
\begin{proof}
    The structure of the proof is identical to that of the proof of Lemma \ref{lem:combred-valid}. 
        
        We introduce generic validity of formulas just as we did for formulas over $(\ptWei, \sqcap, \times)$.
        It is clear that universal validity implies generic validity.

        To show that that generic validity implies combinatorial validity one 
        starts with a Weihrauch reduction $g_t \leq_W g_u$ where 
        $g_t$ and $g_u$ are ``generic interpretations'' of terms $t$ and $u$ 
        respectively.
        By a straightforward induction on the definitions of $\times$, $\sqcap$, $\sqcup$, and $\slice$,
        one obtains that the forward Weihrauch reduction yields a combinatorial reduction from $t$ to $u$,
        whose correctness is witnessed by the backwards reduction.
        To prove the latter part, 
        one queries in $g_t$ the generic Weihrauch degrees associated with different occurrences of non-parallelized variables for pairwise distinct integers.
        The parallelized degrees associated with occurrences of parallelized variables are queried for pairwise distinct integers that are also pairwise distinct from all the queries to non-parallelized variables.
        Further, one ensures that every parallelized degree associated with an occurrence of a parallelized variable is queried for a number of inputs that is strictly larger than the number of occurrences of variables in $u$.
    
        Finally, it is fairly straightforward to construct a Weihrauch reduction from a combinatorial reduction, analogously to Lemma \ref{lem:combred-valid}.
\end{proof}

Lemma \ref{lem: combinatorial and universal validity for terms with join and *} 
allows us to extend a complete axiomatization for 
$(\ptWei, \sqcap, \times, 1)$
to a complete axiomatization for 
$(\ptWei, \sqcap, \times, 1, \sqcup, (-)^*)$
by adding the axioms in Figure \ref{fig:axiom-ext}:

\begin{proposition}\label{Proposition: completeness from completeness}
    Let $\Gamma$ be a complete set of axioms for $(\ptWei, \sqcap, \times, 1)$.
    Let $\Delta$ denote the set of axioms from Figure \ref{fig:axiom-ext}.
    Then $\Gamma \cup \Delta$ is a complete set of axioms for
    $(\ptWei, \sqcap, \times, 1, \sqcup, (-)^*)$.
\end{proposition}
\begin{proof}
    We proceed in stages, starting from the complete axiomatization
    of $(\ptWei, \sqcap, \times)$ and improving to a complete axiomatization
    of $(\ptWei, \sqcap, \times, 1, \sqcup, (-)^*)$ by claiming that $\Gamma \cup \Delta$
    is a complete axiomatization of the substructures obtained by adding the
    connectives $(-)^*$, $\sqcup$ and $1$ successively.
    \subparagraph*{Adding parallelization:} Assume we have $t \lept u$ universally
      valid in $(\ptWei, \sqcap, \times, (-)^*)$. Without loss of generality,
      we can assume that $(-)^*$ is only applied to variables in $t$ and $u$.
    Then, by Lemma \ref{lem: combinatorial and universal validity for terms with join and *}, there exists a combinatorial reduction from $t$ to $u$.
    Since both terms are $\sqcup$-free, this reduction is given by a single relation 
    $f_t \subseteq V_u \times V_t$ 
    as in Definition \ref{def:comb-valid-ext}.
    Consider those vertices $x \in V_u$ that are related by $f_t$ to a set of vertices $\{y_1,\dots,y_m\}$ in $t$ with $m > 1$.
    Such vertices necessarily correspond to subterms of $u$ of the form $a^*$, where $a$ is a variable.
    Construct a new graph $G'$ from $\interpG{u}$ by replacing all such vertices $x$ in $\interpG{u}$ by a clique of $m$ vertices $x_1,\dots, x_m$, where $m$ is the size of the image of $\{x\}$ under $f_t$, all coloured with the same colour as $x$ and with $p(x_i) = 1$ for all $i$.
    We have $G' = \interpG{u'}$, where $u'$ is obtained from $u$ by replacing certain subterms of the form $a^*$ with subterms of the form $a^* \times \dots \times a^*$.
    The equality $u \eqpt u'$ is derivable in $\Gamma \cup \Delta$ by employing relevance of $\times$ and the duplication axiom from Figure \ref{fig:axiom-ext}.
    By construction, the combinatorial reduction $f_t$ from $t$ to $u$ can be made into a single-valued combinatorial reduction $g$ from $t$ to $u'$.
    Fix a mapping $P \colon \mathcal{V} \to \mathcal{V}$ between variables, sending the variables that occur in $u$ or $t$ to fresh distinct variables, not occurring in either term.
    Let $\hat{t}$ denote the term which is obtained from $t$ by replacing all subterms of $t$ of the form $a^*$ for a variable $a$ with the variable $P(a)$.
    Let $\hat{u}$ denote the term which is obtained from $u'$ by replacing all subterms of $u'$ of the form $a^*$ for a variable $a$ as follows: 
    If the vertex in $\interpG{u'}$ corresponding to the subterm gets mapped by $g$ to a vertex which corresponds to a term of the form $a^*$, replace the subterm with the variable $P(a)$.
    If the vertex $\interpG{u'}$ corresponding to the subterm gets mapped by $g$ to a vertex which corresponds to a term of the form $a$, replace the term with the variable $a$.
    By construction, the terms $\hat{u}$ and $\hat{t}$ belong to the language of the fragment $(\ptWei, \sqcap, \times, 1)$,
    and we have a combinatorial reduction from $\hat{t}$ to $\hat{u}$.
    By assumption, this yields a derivation over $\Gamma$ of the inequality $\hat{t} \lept \hat{u}$.
    By reversing the substitution we have performed everywhere in the derivation, we obtain a derivation of the inequality
    $t \lept u'$. Since $u' \eqpt u$ is derivable, so is $t \lept u$ as required.

  \subparagraph*{Adding joins:}
    Now, consider the full fragment $(\ptWei, \sqcap, \times, \sqcup, (-)^*)$.
    Suppose that an inequality $t \lept u$ in this fragment is universally valid.
    Then by the above considerations we obtain a derivably equivalent formula of the form
    $\bigsqcup_i t_i \lept \bigsqcup_j u_j$
    where the $t_i$s and $u_j$s are $\sqcup$-free and finite parallelizations are applied only to variables.
    By the axiom stating that $\sqcup$ is the least upper bound, this further reduces to deriving 
    the universally valid formula
    $t_i \lept \bigsqcup_j u_j$ 
    for all $i$.

    Now, observe that 
    $t \lept u \sqcup v$
    is universally valid for a $\sqcup$-free term $t$ as above if and only if 
    $t \lept u$ or $t \lept v$ is universally valid.
    Indeed, by Lemma \ref{lem: combinatorial and universal validity for terms with join and *},
    universal validity of an inequality is equivalent to the existence of a combinatorial reduction.
    By definition, a combinatorial reduction from $t$ to $u \sqcup v$ for $\sqcup$-free $t$ immediately yields a combinatorial reduction from $t$ to $u$ or from $t$ to $v$, since we have
    $\slice(t) = \{t\}$
    and
    $\slice(u \sqcup v) = \slice(u) \cup \slice(v)$

    Hence, from the universal validity of 
    $t_i \lept \bigsqcup_j u_j$ 
    we obtain the universal validity of 
    $t_i \lept u_j$ for some $j$, which is over $(\ptWei, \sqcap, \times, (-)^*)$
    and thus derivable from $\Gamma \cup \Delta$.
    
    \subparagraph*{Adding $1$:}
    Given a term $t$ over the full signature, it is easy to check whether it is
    equivalent to $1$ and then prove from $\Gamma \cup \Delta$ that it is
    equivalent to it if need be, or to derive $t \eqpt t'$ with $t'$ free of $1$.
    Then if we are to prove a universally valid inequality $t \lept u$ we have three cases:
    \begin{itemize}
     \item both $t$ and $u$ are equivalent to $1$-free terms hence we derive $t \eqpt t' \lept u' \eqpt u$ using our previous point
     \item if $t \eqpt 1$, then we can then derive $t \eqpt 1 \lept u$ where the latter is derived by recurring over $u$
     \item if $u \eqpt 1$, then it is also the case that $t \eqpt 1$ since $1$ is the bottom element of $\ptWei$, and
       we can derive $u \eqpt 1 \eqpt t$.
    \end{itemize}

\end{proof}

Finally, we observe that we can extend our notion of combinatorial validity to non-pointed terms over $(\Wei, \sqcap, \times, (-)^*,1)$ by requiring that the relations $f_{t'}$ be left-total and handling $1$ as indicated at the end of Section \ref{section:combinatorial}.
Using this, we obtain an analogous result to \Cref{Proposition: completeness from completeness} for the non-pointed case.
The proof is almost identical, so we will not spell it out in detail.

\section{Complexity of deciding validity}

We turn to the problem of determining the validity of an inequality over $(\ptWei, \sqcap, \times, 1, (-)^*)$ or 
$(\ptWei, \sqcap, \times, 1, \sqcup, (-)^*)$.
We first consider the $\sqcup$-free case:

\begin{theorem}\label{thm: join-free case is Sigma2p complete}
The problem ``is the inequality $t \lept u$ valid?'' in the structure $(\ptWei, \sqcap, \times)$ is $\Sigma^p_2$-complete.
\end{theorem}
\begin{proof}
In view of~\Cref{lem: combinatorial and universal validity for terms with join and *} we can consider the equivalent problem of deciding the existence of a combinatorial reduction from $t$ to $u$, \textit{i.e.}~a $\bullet$-reduction from $\interpG{t}$ to $\interpG{u}$.

For containment in $\Sigma_2^p$, note that the existence of a reduction can be essentially written as
\[
\exists f : V_u \partto V_t. ~~
\underbrace{\forall~ C \text{ max. clique of } \interpG{u}. ~~ \overbrace{f(C) \text{ contains a max. clique of } \interpG{u}}^{\text{polytime}}}_{\text{coNP}}.
\]
In the above, all quantifiers involve polynomial-sized objects relative to $G$ and $H$. 
It hence suffices to count quantifiers and argue that we may indeed check the inner condition in polynomial time.

In general, checking whether a set of vertices contains a maximal clique may be hard.
However, in our case, since we are working with denotations of terms, we are not working with general graphs, but cographs.
In this case, there is a straightforward algorithm for deciding this problem.
For a term $t$, let the predicate $\operatorname{MClique}_{t}(s, V)$ for subterms $s$ of $t$ and sets $V$ of occurrences of variables in $t$ be defined as follows:
\begin{enumerate}
    \item For an occurrence $v$ of a variable,  
    $\operatorname{MClique}_{t}(v, V)$
    holds true if and only if $v \in V$.
    \item $\operatorname{MClique}_{t}(s \sqcap u, V) = \operatorname{MClique}_{t}(s, V) \lor \operatorname{MClique}_{t}(u, V)$.
    \item $\operatorname{MClique}_{t}(s \times u, V) = \operatorname{MClique}_{t}(s, V) \land \operatorname{MClique}_{t}(u, V)$.
\end{enumerate}
If we represent the set $V$ appropriately, the predicate $\operatorname{MClique}_{t}(s, V)$ can be evaluated in essentially linear time.
It is clear that $\operatorname{MClique}_{t}(s, V)$ holds true if and only if the set $V$ contains a maximal clique in $\interpG{t}$.

To show $\Sigma_2^p$-hardness, we reduce in polynomial time TQBF restricted to $\Sigma_2$ formulas to our problem. 
Given a quantified boolean formula
\[\exists x_1, \ldots, x_k \forall y_1, \ldots, y_n ~ \bigwedge_{i \in I} \bigvee_{j \in J_i} \ell_j, \]
with the $\ell_j$s being literals generated from $\{x_1, \ldots, x_k, y_1, \ldots, y_n\}$,
define $K = 1 + \sum\limits_{i \in I} |J_i|$. 
We may then produce in polynomial time the following terms whose variables are literals (where powers are to be understood as iterated $\times$):
\begin{itemize}
\item $E_L = \bigtimes\limits_{m = 1}^k x_m^K \sqcap \overline{x_m}^K$ and $E_R = \bigtimes\limits_{m = 1}^k x_m^K \times \overline{x_m}^K$
\item $A = \bigtimes\limits_{m = 1}^n y_m^K \sqcap \overline{y_m}^K$
\item $F = \bigtimes\limits_{i \in I} \bigsqcap\limits_{j \in J_i} \ell_{j}$
\end{itemize}
To conclude, we only need to argue that the formula is satisfiable if and only if the inequality
\[ E_L \times F \lept E_R \times A\]
is valid.

Assume that the formula is satisfiable.
Then there exists a valuation $\rho \colon \{x_1,\dots,x_k\} \to \{0,1\}$ of the $x_i$s such that for all valuations 
$\theta \colon \{y_1,\dots,y_n\} \to \{0,1\}$ of the $y_j$s, the combined valuation $(\rho, \theta)$ renders the formula true.

Pick any $\rho$ as above. 
Then there exists a partial mapping $f \colon E_R \times A \partto E_L \times F$ with the following properties:
\begin{enumerate}
    \item If $\rho(x_i) = 1$, the map $f$ sends all occurrences of 
    $\overline{x_i}$ in $E_R$ to an occurrence of 
    $\overline{x_i}$ in $E_L$, 
    such that the clique 
    $\overline{x_i}^K$ in $E_L$ is contained in the image of $f$ under he clique 
    $\overline{x_i}^K$ in $E_R$.
    Further, the image of the clique $x_i^K$ in $E_R$ covers all occurrences of $x_i$ in $F$.
    \item If $\rho(x_i) = 0$, the map $f$ has the same properties, with the roles of $x_i$ and $\overline{x_i}$ reversed.
    \item The image of the clique $y_m^K$ covers all occurrences of $y_m$ in $F$.
    \item The image of the clique $\overline{y_m}^k$ covers all occurrences of $\overline{y_m}$ in $F$.
\end{enumerate}
The mapping $f$ exists essentially because $K$ is chosen large enough to ensure that we can cover all occurrences of variables in $F$ as described above.

 Now, we claim that $f$ is a combinatorial reduction from $E_L \times F$ to $E_R \times A$.
 Let $C$ be maximal clique in $E_R \times A$.
 Then $C$ is of the form $C = E_R \times C'$, where $C'$ is a maximal clique in $A$.
 Further, $C'$ is of the form
 \[
     C' = \bigtimes\limits_{m = 1}^n z_m^K
 \] 
 where $z_m \in \{y_m, \overline{y_m}\}$.
 Let 
     $\theta \colon \{y_1,\dots,y_n\} \to \{0,1\}$
 be the valuation that sends $y_j$ to $1$ if $z_j = y_j$ and to $0$ if $z_j = \overline{y_j}$.
 Then, by assumption, some conjunction of literals 
 $\bigwedge_{i \in I} \ell_{j(i)}$ 
 holds true under the valuations $\rho$ and $\theta$.
 Consider the corresponding maximal clique 
 $\bigtimes_{i \in I} \ell_{j(i)}$
 in $F$.
 Let $E_L'$ be the maximal clique $\bigtimes_{m = 1}^k w_m^K$ of $E_L$ where 
 $w_m \in \{x_m, \overline{x_m}\}$ 
 and 
 $w_m = x_m$
 if and only if $\rho(x_m) = 0$.
 Then $E_L' \times \bigtimes_{i \in I} \ell_{j(i)}$ is a maximal clique in $E_L \times F$.
 By construction, every vertex in $E_L$ is the image of a vertex in $E_R$ under $f$.

 Now, by assumption, $\ell_{j(i)}$ is consistent with $\rho$ and $\theta$ in the sense that 
 if $\ell_{j(i)} \in \{x_m, \overline{x_m}\}$, then $\ell_j(i) = x_m$ if and only if $\rho(x_m) = 1$
 and similarly for $\theta$ and  $\ell_{j(i)} \in \{y_m, \overline{y_m}\}$.
 But this implies by construction that every $\ell_{j(i)} \in \{y_m, \overline{y}_m\}$ is the image of 
 a vertex in $C'$ under $f$ and every $\ell_{j(i)} \in \{x_m, \overline{x}_m\}$ is the image of a vertex
 in $E_R$ under $f$.
 This shows that $f$ is a combinatorial reduction.

Conversely, assume that a combinatorial reduction exists between 
$E_L \times F$ 
and 
$E_R \times A$,
witnessed by a partial map 
\[
    f\colon \interpG{E_R \times A} \partto \interpG{E_L \times F}.
\]
Observe that for all $m \in \{1,\dots,k\}$, either $f$ maps the clique $x_m^K$ of $\interpG{E_R}$ into $\interpG{E_L}$, or it maps the clique $\overline{x_m}^K$ of $\interpG{E_R}$ into $\interpG{E_L}$.

Hence, we can define a valuation $\rho \colon \{x_1,\dots,x_k\} \to \{0,1\}$ by letting $\rho(x_i) = 1$ if and only if the clique 
$\overline{x_m}^K$ of $\interpG{E_R}$ is mapped into $\interpG{E_L}$. It is now quite straightforward to check that $\rho$ is a satisfying assignment for the formula.
\end{proof}

Our argument for containment in $\Sigma_2^p$ can easily be adapted to formulas with parallelization applied to variables.
Using the observation that any formula with parallelization can be re-written in this shape in linear time, we obtain:

\begin{corollary}\label{cor: adding parallelization}
    The problem ``is the inequality $t \lept u$ valid?'' in the structure $(\ptWei, \sqcap, \times, (-)^*)$ is $\Sigma^p_2$-complete.
\end{corollary}

The proof of Theorem \ref{thm: join-free case is Sigma2p complete} 
(together with the observation underlying Corollary \ref{cor: adding parallelization})
yields that validity of inequalities in 
$(\ptWei, \sqcap, \times, 1, \sqcup, (-)^*)$ 
is contained in $\Pi_3^p$ --- the additional universal quantifier being used to quantify over the slices.
This additional quantifier alternation cannot be avoided:

\begin{theorem}
    The problem ``is the inequality $t \lept u$ valid?'' in the structure $(\ptWei, \sqcap, \times, 1, \sqcup)$ is
    $\Pi^p_3$-complete.
\end{theorem}
\begin{proof}[Proof idea]
Inclusion in $\Pi^p_3$ is straightforward, following the proof of Theorem \ref{thm: join-free case is Sigma2p complete}.
    For the hardness part, we employ a similar construction to that in the proof of Theorem \ref{thm: join-free case is Sigma2p complete}.
    Given a quantified boolean formula
    \[\forall z_1, \ldots, z_l \exists x_1, \ldots, x_k \forall y_1, \ldots, y_n ~ \bigwedge_{i \in I} \bigvee_{j \in J_i} l_j\]
    we construct the inequality
    \[ Z_L \times E_L \times F \lept Z_R \times E_R \times A\]
    where $E_L$, $F$, $E_R$ and $A$ are taken as in the proof of Theorem \ref{thm: join-free case is Sigma2p complete} ($F$ might contain literals $z_i$ and $\overline{z_i}$), and where
    \[Z_L = \bigtimes_{i = 1}^l (z_i \sqcup \overline{z_i}) \qquad \text{and} \qquad
    Z_R = \bigtimes_{i = 1}^l (z_i^K \sqcup \overline{z_i}^K).\]
\end{proof}

\section{Towards a proof of completeness}

We do not know whether our axiom system yields a complete axiomatization of $(\ptWei, \sqcap, \times, 1)$.
As a first possible step towards a potential completeness proof, we show that we can restrict our attention to inequalities
$t \lept u$ where $t$ and $u$ each have pairwise distinct variables and the combinatorial reduction witness 
is a total bijective map.

We use this to show that our axioms allow us to derive all inequalities of the form $t \lept u$ where $u$ is 
\emph{square-free}.
Call a term square-free\footnote{The name comes from the fact that the graph induced by $\interpG{-}$ has no induced subgraph which is a square.} if, up to
associativity and commutativity, it can be rewritten in the following grammar, where $x$ ranges over variables:
\[ t, t' \bnfeq x \bnfalt t \times x \bnfalt 1 \bnfalt t \sqcap t'\]

\begin{proposition}\label{proposition : reduction wlog total bijective}
    Let $t$ and $u$ be terms over $(\ptWei, \sqcap, \times, 1)$.
    Assume that $t \lept u$.
    Then there exist terms $t'$ and $u'$ such that 
    $t \lept t' \lept u' \lept u$,
    the inequalities $t \lept t'$ and $u' \lept u$ are derivable from the axioms in Figure \ref{fig:axiom-base}, 
    and such that there exists a (total) bijective combinatorial reduction from $t'$ to $u'$.
\end{proposition}
\begin{proof}
    As argued at the end of Section \ref{section:combinatorial}, we can assume without limitation of generality that $t$ and $u$ are either $1$ or do not contain a $1$. If either term is $1$, the claim is trivial. Otherwise, there exists a combinatorial reduction $f \colon V_u \partto V_t$.
    
    Consider a vertex $v$ in $\interpG{t}$ not in the range of $f$.
    Consider the path $\langle o_1,\dots,o_m, v \rangle$ in the (tree associated with the) term $t$ leading to $v$.
    If $o_1 = \dots = o_m = \times$, then every maximal clique in $\interpG{t}$ contains $v$, contradicting the assumption that $f$ is a combinatorial reduction.
    Hence, there exists a maximal $j \in \{1,\dots,m\}$ with $o_j = \sqcap$.
    Let $t'$ be the term which is obtained from $t$ by removing the sub-tree with root $o_{j + 1}$ (letting $o_{m + 1} = v$).
    Then we can derive $t \lept t'$ since meets are lower bounds.
    The maximal cliques of $\interpG{t'}$ are the maximal cliques of $t$ that do not contain $v$.
    Since no maximal clique that contains $v$ contributes to $f$ being a combinatorial reduction, the co-restriction of $f$ to 
    $V_{t'}$ remains a combinatorial reduction.
    We can apply the above construction to remove all vertices outside the range of $f$ to ensure that $f$ is a surjective partial map -- potentially with smaller domain.
    
    If $f$ is not injective, say $f^{-1}(\{v\}) = \{w_1,\dots,w_m\}$, we can replace $v$ in $\interpG{t}$ with the $m$-fold meet $v \sqcap \dots \sqcap v$ (coloured with the same colour as $v$) and send each $w_i$ to a separate occurrence of $v$.
    It is straightforward to check that this yields a combinatorial reduction from $u$ to a term $t'$ which is derivably equivalent to $t$.
    
    In the case where $f$ is not total with $v \in V_u \setminus \dom(f)$,
    we consider the path $\langle o_1,\dots,o_m, v \rangle$ in the (tree associated with the) term $u$ leading to $v$.
    There must exist a maximal $j \in \{1,\dots,m\}$ with $o_j = \times$.
    Removing in $u$ the the sub-tree with root $o_{j + 1}$ (letting $o_{m + 1} = v$) yields a derivably smaller term $w$
    such that $f$ is still a combinatorial reduction from $\interpG{w}$ to $\interpG{t}$.
    Applying this construction to all vertices outside $\dom(f)$, we obtain a term $u'$ as claimed.
\end{proof}

Proposition \ref{proposition : reduction wlog total bijective} implies that a complete axiomatization for formulas of the form $t \lept u$ where $t$ and $u$ both have pairwise distinct variables yields a complete axiomatization of the whole theory: 
if $t$ and $u$ are arbitrary terms such that $t \lept u$ is universally valid, there exists without loss of generality a bijective combinatorial reduction from $u$ to $t$. This allows us to transport an injective colouring of the vertices of $\interpG{u}$ to an injective colouring of the vertices of $\interpG{t}$, yielding a valid and hence derivable inequality $t' \lept u'$ between terms with pairwise distinct variables. 
Substituting back the original variable names in the derivation of $t' \lept u'$ yields a derivation of $t \lept u$.

\begin{proposition}
Our axiom system is complete for inequalities $t \lept u$ where $u$ is square-free.
\end{proposition}
\begin{proof}
    By the remark following Proposition \ref{proposition : reduction wlog total bijective} we may assume that the variables are pairwise distinct in both $t$ and $u$. We show the claim by induction over $u$. 
If $u$ is a variable, the constant $1$, or of the form $u_0 \sqcap u_1$, this is obvious. 

The remaining case is $u = u' \times x$.
An easy induction over $t$ shows that we can derive $t \lept t[1/x] \times x$.
Since $x$ does not appear in $u'$, we must have $t[1/x] \lept u'$.
By the induction hypothesis, this last inequality is derivable from our axioms.
Now, combining this derivation with monotonicity of $\times$ yields a derivation of 
$t[1/x] \times x \lept u' \times u$
which in total yields a derivation of $t \lept u$.
\end{proof}

\section*{Acknowledgments}
We are grateful to Manlio Valenti and {\fontencoding{T5}\selectfont{}Lê Thành Dũng Nguyễn} for discussions, and also to Matteo Acclavio for some pointers to the literature.

\bibliographystyle{eptcs}
\bibliography{biblio}

\end{document}